\documentclass{article}

\usepackage{amsmath}
\usepackage{amssymb}

\title{Constructing the space of valuations of a quasi-Polish space as a space of ideals}

\author{Matthew de Brecht \thanks{This work was supported by JSPS KAKENHI Grant Number 18K11166. We thank the reviewers for carefully reading this paper and providing feedback.}}

\date{}

\newtheorem{theorem}{Theorem}
\newtheorem{definition}[theorem]{Definition}
\newtheorem{lemma}[theorem]{Lemma}

\newenvironment{proof}{\noindent{\bf Proof:}}{\qed\medskip}
\def\squareforqed{\hbox{\rlap{$\sqcap$}$\sqcup$}}
\def\qed{\ifmmode\squareforqed\else{\unskip\nobreak\hfil
\penalty50\hskip1em\null\nobreak\hfil\squareforqed
\parfillskip=0pt\finalhyphendemerits=0\endgraf}\fi}
\def\IN{{\mathbb N}}
\def\IQ{{\mathbb Q}}
\def\sierp{{\mathbb S}}
\def\Rext{\overline{\mathbb R}_+}

\def\calB{{\mathcal B}}
\def\calP{{\mathcal P}}

\def\bpi#1{\mathbf{\Pi}^0_{#1}}
\def\coanalytic{\mathbf{\Pi}^1_1}

\def\PU{{\mathrm{\mathbf K}}}
\def\PO{{\mathrm{\mathbf O}}}
\def\PV{{\mathrm{\mathbf V}}}

\def\I#1{{{\mathrm{\mathbf I}}(#1)}}

\def\calPfin{{\mathcal P}_{\mathrm{fin}}}

\newcommand{\QCBZ}{\mathsf{QCB_0}}
\newcommand{\ACAZ}{\mathsf{ACA_0}}
\newcommand{\coanalyticCA}{{\mathbf{\Pi}^1_1-\mathsf{CA_0}}}
\newcommand{\uparw}{{\uparrow}}

\begin{document}

\maketitle

\begin{abstract}
We construct the space of valuations on a quasi-Polish space in terms of the characterization of quasi-Polish spaces as spaces of ideals of a countable transitive relation. Our construction is closely related to domain theoretical work on the probabilistic powerdomain, and helps illustrate the connections between domain theory and quasi-Polish spaces. Our approach is consistent with previous work on computable measures, and can be formalized within weak formal systems, such as subsystems of second order arithmetic.
\end{abstract}

%%%%%%%%%%%%%%%%%%%%%%%%%%%%%%%%%%%%%%%%%%%%%%%%%%
%%%%%%%%%%%%%%%%%%%%%%%%%%%%%%%%%%%%%%%%%%%%%%%%%%
%%%%%%%%%%%%%%%%%%%%%%%%%%%%%%%%%%%%%%%%%%%%%%%%%%

\section{Introduction}

Quasi-Polish spaces \cite{dbr} are a class of well-behaved countably based sober spaces that includes Polish spaces, $\omega$-continuous domains, and countably based spectral spaces. They can be interpreted via Stone-duality as the spaces of models of countably axiomatized propositional geometric theories \cite{H15,Chen19}. In \cite{DPS} another characterization of quasi-Polish spaces was presented that is a natural generalization of the notion of an \emph{abstract basis} for $\omega$-continuous domains \cite{g03}. In this paper we use this latter characterization to extend domain theoretical work on probabilistic powerdomains to the study of valuations on quasi-Polish spaces.

Valuations are a substitute for Borel measures which are used in the denotational semantics of probabilistic programming languages \cite{jones:phd} and in computable approaches to measure theory, probability theory, and randomness \cite{Sch:Prob,HR09,PSZ20}. See R.~Heckmann's excellent paper \cite{H96} for more on the theory of valuations, spaces of valuations, and integration\footnote{The valuations in this note correspond to the \emph{Scott-continuous valuations} in \cite{H96}.}. Every valuation on a quasi-Polish space can be extended to a Borel measure \cite{DGJL}, and this extension is unique if the valuation is locally finite \cite{dbrCCA15}. Conversely, it is easy to see that the restriction of a Borel measure to the open sets is a valuation. Thus, in particular, there is a bijection between probabilistic valuations and probabilistic Borel measures on quasi-Polish spaces.

The main result in this paper is a construction of the space of valuations on a quasi-Polish space as a space of ideals of a transitive relation on a countable set (Theorem~\ref{thrm:mainresult}). Our construction is closely related to domain theoretical work on the probabilistic powerdomain (see \cite{jones:phd} and \cite[Section~IV-9]{g03}). Along with the constructions of the upper and lower powerspaces of quasi-Polish spaces as spaces of ideals given in \cite{dbr20}, our results demonstrate how some domain theoretic results generalize well to quasi-Polish spaces (see also \cite{DK19} for more on the upper and lower powerspaces of quasi-Polish spaces). 

An immediate corollary of our construction is that the space of valuations on a quasi-Polish space is again a quasi-Polish space, although this already follows from well-known results. A locale theoretic proof easily follows from S.~Vickers' geometricity result in \cite[Proposition~5]{V08} by using R.~Heckmann's characterization of quasi-Polish spaces as countably presented locales \cite{H15}. A proof based on quasi-metrics, at least for the case of subprobabilistic valuations, follows from J.~Goubault-Larrecq's work on continuous Yoneda-complete quasi-metric spaces in \cite[Section~11]{GL17} and his characterization of quasi-Polish spaces in \cite[Theorem~8.18]{GL17b}. Independently, the first proof we found (which we presented at the Domains~XII conference in August 2015) was largely based on M.~Schr\"{o}der's work in \cite{Sch:Prob} on the space of (probabilistic) measures within the cartesian closed category $\QCBZ$. That proof starts with the observation that the $\QCBZ$ exponential $\sierp^{\sierp^X}$ is quasi-Polish whenever $X$ is\footnote{See \cite{DK19} for a proof. The $\sierp$ here is the Sierpinski space, and the space $\PO(\PO(X))$ defined in \cite{DK19} is homeomorphic to the $\QCBZ$ exponential object $\sierp^{\sierp^X}$ when $X$ is quasi-Polish.}, then uses the cartesian closed structure of $\QCBZ$ to show that $Y^{\sierp^X}$ is quasi-Polish whenever $X$ and $Y$ are, and finally observes that M.~Schr\"{o}der's construction of the space of valuations on $X$ can be obtained as the equalizer of the continuous functions $\ell, r\colon \Rext^{\sierp^X} \to\Rext\times\Rext^{\sierp^X\times \sierp^X}$ defined as:
\begin{eqnarray*}
\ell(\nu) &=& \big\langle \nu(\emptyset), \lambda\langle U,V\rangle. \nu(U)+\nu(V) \big\rangle, \text{ and}\\
r(\nu) &=& \big\langle 0, \lambda\langle U,V\rangle. \nu(U\cup V)+\nu(U\cap V) \big\rangle.
\end{eqnarray*}
It follows that the space of valuations is quasi-Polish because the space of extended reals $\Rext$ is quasi-Polish and the category of quasi-Polish spaces is closed under countable limits.

A nice characteristic of the construction we give in this paper is that it can be formalized within relatively weak formal systems. For example, our approach is related to C.~Mummert's formalization of general topology within subsystems of second order arithmetic \cite{mummert:phd,M06,MS10}\footnote{Note that C.~Mummert's \emph{MF-spaces} are in general $\coanalytic$-complete spaces, whereas quasi-Polish spaces correspond to the $\bpi 2$-level of the Borel hierarchy. This explains why $\coanalyticCA$ is required to prove MF-spaces are closed under $G_\delta$-subsets, whereas our construction of $\bpi 2$-subspaces of quasi-Polish spaces in Theorem~3 of \cite{dbr20} can be done within $\ACAZ$.}.

\section{Main result}

We let $\Rext$ denote the positive extended reals (i.e., $[0,\infty]$) with the Scott-topology induced by the usual order. Given a topological space $X$, we let $\PO(X)$ denote the lattice of open subsets of $X$ with the Scott-topology. 

\begin{definition}[Valuations]
Let $X$ be a topological space. A \emph{valuation} on $X$ is a continuous function $\nu\colon \PO(X) \to \Rext$ satisfying:
\begin{enumerate}
\item
$\nu(\emptyset) = 0$, and\hfill(\emph{strictness})
\item
$\nu(U) + \nu(V) = \nu(U\cup V) + \nu(U \cap V)$.\hfill(\emph{modularity})
\end{enumerate}
The \emph{space of valuations} on $X$ is the set $\PV(X)$ of all valuations on $X$ with the \emph{weak topology}, which is generated by subbasic opens of the form
\[ \langle U,q \rangle := \{\nu\in\PV(X) \mid \nu(U) > q\}\]
with $U \in \PO(X)$ and $q \in \Rext\setminus\{\infty\}$.
\qed
\end{definition}

In this paper we will only consider the whole space of valuations $\PV(X)$, but it is straightforward to modify our results for the subspaces of $\PV(X)$ consisting of  \emph{probabilistic valuations} (i.e., valuations satisfying $\nu(X)=1$) and \emph{sub-probabilistic valuations} (i.e., valuations satisfying $\nu(X) \leq 1$). 

Quasi-Polish spaces were introduced in \cite{dbr}. In this paper we will define them using the following equivalent characterization from \cite{DPS} (see also \cite{dbr20}).
\begin{definition}\label{def:idealspace}
Let $\prec$ be a transitive relation on $\IN$. A subset $I\subseteq \IN$ is an \emph{ideal} (with respect to $\prec$) if and only if:
\begin{enumerate}
\item
$I \not=\emptyset$,\hfill (\emph{$I$ is non-empty})
\item
$(\forall a \in I) (\forall b \in \IN)\, (b \prec a \Rightarrow b \in I)$,\hfill (\emph{$I$ is a lower set})
\item
$(\forall a,b \in I)(\exists c\in I)\, (a \prec c  \,\&\,  b \prec c)$.\hfill (\emph{$I$ is directed})
\end{enumerate}
The collection $\I{\prec}$ of all ideals has the topology generated by basic open sets of the form $[n]_{\prec} = \{ I \in \I{\prec} \mid n \in I\}$. A space is \emph{quasi-Polish} if and only if it is homeomorphic to $\I{\prec}$ for some transitive relation $\prec$ on $\IN$.
\qed
\end{definition}
We often apply the above definition to other countable sets with the implicit assumption that it has been suitably encoded as a subset of $\IN$.

Fix a transitive relation $\prec$ on $\IN$ for the rest of this section. Let $\calB$ be the (countable) set of all partial functions $r:\subseteq \IN \to \IQ_{>0}$ such that $dom(r)$ is finite, where $\IQ_{>0}$ is the set of rational numbers strictly larger than zero.

\begin{definition}
Define the transitive relation $\prec_V$ on $\calB$ as $r \prec_V s$ if and only if 
\[\sum_{b\in F} r(b) < \sum_{c \in \uparw F \cap dom(s)} s(c)\]
for every non-empty $F\subseteq dom(r)$, where $\uparw F = \{ c\in\IN \mid (\exists b\in F)\, b\prec c \}$.
\qed
\end{definition}

Transitivity of $\prec_V$ follows from the transitivity of $\prec$. Note that if $dom(r)=\emptyset$ then $r \prec_V s$ for every $s \in \calB$. We will sometimes use the fact that if $r \prec_V s$ and $b \in dom(r)$ then there is $c \in dom(s)$ with $b \prec c$.

\begin{definition}
Define $f_V\colon \PV(\I{\prec}) \to \I{\prec_V}$ and $g_V \colon \I{\prec_V} \to \PV(\I{\prec})$ as 
\begin{eqnarray*}
f_V(\nu) &=& \left\{ r\in \calB \,\middle|\,  \sum_{b\in F} r(b) < \nu(\bigcup_{b\in F} [b]_{\prec}) \text{ for every non-empty $F\subseteq dom(r)$} \right\},\\
g_V(I) &=& \lambda U.\bigvee\left\{ \sum_{b\in dom(r)} r(b) \,\middle|\, r\in I \text{ and } \bigcup_{b\in dom(r)} [b]_{\prec}\subseteq U \right\}.
\end{eqnarray*}
\qed
\end{definition}

We next prove a few lemmas which will be used to show that $f_V$ and $g_V$ are continuous inverses of each other.

\begin{lemma}\label{lem:formval_restrictions}
If $I \in \I{\prec_V}$, $r \in I$, and $A \subseteq dom(r)$, then $r|_A \in I$, where $r|_A$ is the partial function obtained by restricting the domain of $r$ to $A$.
\end{lemma}
\begin{proof}
Since $I$ is directed there is $s\in I$ with $r\prec_V s$. Then clearly $r|_A \prec_V s$ hence $r|_A \in I$ because $I$ is a lower set.
\end{proof}

\begin{definition}
Define the transitive binary relation $\prec_U$ on $\calPfin(\IN)$ (the set of finite subsets of $\IN$) as $F \prec_U G$ if and only if $(\forall n\in G)\,(\exists m\in F)\, m \prec n$.
\qed
\end{definition}

We write $\PU(X)$ for the space of saturated compact subsets of  $X$ (see \cite{DK19}).

\begin{lemma}[Lemma~9 \& Theorem~10 of \cite{dbr20}]\label{lem:upperpowerspacelemma}
Given $J \in \I{\prec_U}$, the set 
\[g_U(J) = \{ I \in \I{\prec} \mid (\forall F\in J)(\exists m\in I)\,m\in F\}\]
is in $\PU(\I\prec)$. Furthermore, for any $S\subseteq \IN$, 
$g_U(J) \subseteq \bigcup_{b\in S}[b]_{\prec}$ if and only if there is finite $F\subseteq S$ with $F \in J$.
\qed
\end{lemma}

\begin{lemma}\label{lem:formval_upperpower}
If $I \in \I{\prec_V}$ and $r \in I$, then there exists $s \in I$ with $r \prec_V s$ and $dom(r) \prec_U dom(s)$.
\end{lemma}
\begin{proof}
Choose any $t\in I$ with $r \prec_V t$. Let $s$ be the restriction of $t$ to have $dom(s) = \{ c\in dom(t) \mid (\exists b\in dom(r))\, b \prec c\}$. Clearly $r \prec_V s$ and $dom(r) \prec_U dom(s)$, and Lemma~\ref{lem:formval_restrictions} implies $s \in I$.
\end{proof}

\begin{lemma}\label{lem:formval_upperpower2}
Assume $I \in \I{\prec_V}$ and $r \in I$. Then there exists $K \in \PU(\I{\prec})$ such that
\begin{itemize}
\item
$K \subseteq \bigcup_{b\in dom(r)} [b]_{\prec}$, and 
\item
For any finite $F \subseteq \IN$, if $K \subseteq \bigcup_{b\in F} [b]_{\prec}$, then there is $s\in I$ with $r \prec_V s$ and $F \prec_U dom(s)$ and $K \subseteq \bigcup_{c\in dom(s)} [c]_{\prec} \subseteq \bigcup_{b\in F} [b]_{\prec}$.
\end{itemize}
\end{lemma}
\begin{proof}
Fix $I \in \I{\prec_V}$ and $r\in I$. Using Lemma~\ref{lem:formval_upperpower}, we can find a $\prec_V$-ascending sequence $(r_i)_{i\in\IN}$ in $I$ with $r=r_0$ and $dom(r_i) \prec_U dom(r_{i+1})$ for each $i\in\IN$. Then $J = \{ F \in \calPfin(\IN) \mid (\exists i\in\IN)\, F \prec_U dom(r_i) \}$ is in $\I{\prec_U}$, hence $K = g_U(J)\in \PU(\I{\prec})$ and $K \subseteq \bigcup_{b\in dom(r)} [b]_{\prec}$ by Lemma~\ref{lem:upperpowerspacelemma} and the fact that $dom(r) \in J$. Assume $F \subseteq \IN$ is finite and $K \subseteq \bigcup_{b\in F} [b]_{\prec}$. Then $F \in J$ by Lemma~\ref{lem:upperpowerspacelemma}, hence $F \prec_U dom(r_i)$ for some $i\in \IN$. Since $\prec_U$ is transitive, we can assume without loss of generality that $i >0$. Setting $s = r_i$, we have $s \in I$ and $r \prec_V s$ and $F \prec_U dom(s)$, and since $dom(s) \in J$ it follows from Lemma~\ref{lem:upperpowerspacelemma} that $K \subseteq \bigcup_{c\in dom(s)} [c]_{\prec}$. The claim $\bigcup_{c\in dom(s)} [c]_{\prec} \subseteq \bigcup_{b\in F} [b]_{\prec}$ follows from $F \prec_U dom(s)$.
\end{proof}

\begin{lemma}\label{lem:sum_of_differences}
Let $D \subseteq \IN$ be finite,  and let $\calP_{+}(D)$ be the set of non-empty subsets of $D$. Define
\begin{eqnarray*}
U_G &=& \bigcap_{b\in G} [b]_{\prec}\\
V_G &=& U_G \cap \bigcup_{b\in D \setminus G} [b]_{\prec}
\end{eqnarray*}
for each $G \in\cal{P}_{+}(D)$. Let $P \subseteq \calP_{+}(D)$ be an upper set (i.e., if $F \in P$ and $F \subseteq G \subseteq D$ then $G\in P$). If $\nu\in\PV(\I{\prec})$ and $\nu(U_G) < \infty$ for each $G \in P$, then 
\[\sum_{G\in P} (\nu(U_G) - \nu(V_G)) = \nu\left(\bigcup_{G\in P} U_G \right).\] 
\end{lemma}
\begin{proof}
The proof is by induction on the size of $P$. It is trivial when $P = \emptyset$, so assume $P$ is a non-empty upper set and that the lemma holds for all upper sets of size strictly less than $P$. If $F$ is any minimal element of $P$, then
\begin{eqnarray*}
V_F &=& \bigcup_{b\in D\setminus F} U_{F\cup\{b\}}\\
 &=& \bigcup_{G\in P\setminus \{F\}} U_{F\cup G}\\
 &=& U_F \cap \bigcup_{G\in P\setminus \{F\}}  U_G,
\end{eqnarray*}
so the induction hypothesis and modularity yields
\begin{eqnarray*}
\sum_{G\in P} (\nu(U_G) - \nu(V_G)) &=& \nu(U_F) - \nu(V_F) + \sum_{G\in P\setminus\{F\}} (\nu(U_G) - \nu(V_G))\\
&=& \nu(U_F) - \nu\left(U_F \cap \bigcup_{G\in P\setminus \{F\}}  U_G\right) +  \nu\left(\bigcup_{G\in P\setminus\{F\}} U_G \right)\\
&=& \nu\left(\bigcup_{G\in P} U_G \right).
\end{eqnarray*}
\end{proof}

\begin{lemma}
$f_V$ is well-defined and continuous.
\end{lemma}
\begin{proof}
We first show that $f_V(\nu) \in \I{\prec_V}$ for each $\nu\in\PV(\I{\prec})$.
\begin{enumerate}
\item
\emph{($f_V(\nu)$ is non-empty)}. The partial function with empty domain is in $f_V(\nu)$.
\item
\emph{($f_V(\nu)$ is a lower set)}. Assume $r \prec_V s \in f_V(\nu)$. Let $F \subseteq dom(r)$ be non-empty, and define $G = \uparw F \cap dom(s)$. Since $b \prec c$ implies $[c]_{\prec}\subseteq [b]_{\prec}$ it follows that $\bigcup_{c\in G} [c]_{\prec} \subseteq \bigcup_{b\in F} [b]_{\prec}$. Then
\begin{eqnarray*}
\sum_{b\in F} r(b) &<& \sum_{c \in G} s(c)\qquad\text{\hfill (because $r\prec_V s$)}\\
  &<& \nu(\bigcup_{c\in G} [c]_{\prec}) \qquad\text{\hfill(because $s\in f_V(\nu)$)}\\
  &\leq& \nu(\bigcup_{b\in F} [b]_{\prec}) \qquad\text{\hfill(because $\nu$ is monotonic)},
\end{eqnarray*}
hence $r\in f_V(\nu)$.
\item
\emph{($f_V(\nu)$ is directed)}. Our proof is related to the series of lemmas leading up to Theorem~IV-9.16 in \cite{g03}. Assume $r_0,r_1 \in f_V(\nu)$. For each $i\in\{0,1\}$ and non-empty $F \subseteq dom(r_i)$ fix some real number $\beta^i_F$ satisfying 
\[\sum_{b\in F} r_i(b) < \beta^i_F < \nu\left(\bigcup_{b\in F} [b]_{\prec}\right),\]
and set
\[ \beta = \min\left\{ \frac{\beta^i_F-\sum_{b\in F} r_i(b) }{\sum_{b\in F} r_i(b)} \,\middle|\, i\in\{0,1\} \text{ \& }\emptyset \not= F \subseteq dom(r_i) \right\}. \]
Then $\alpha = 1/\left(1+\beta/2\right)$ satisfies $0 < \alpha < 1$ and is such that 
\[\sum_{b\in F} r_i(b) < \alpha \nu\left(\bigcup_{b\in F} [b]_{\prec}\right)\]
for each $ i\in\{0,1\}$ and non-empty $F \subseteq dom(r_i)$ (see Lemma~IV-9.11 (iii) of \cite{g03}). Set $M = 1 + \sum_{b\in dom(r_0)} r_0(b) + \sum_{b\in dom(r_1)} r_1(b)$, and $D = dom(r_0)\cup dom(r_1)$. Let $U_G$ and $V_G$ be defined as in Lemma~\ref{lem:sum_of_differences} for each non-empty $G \subseteq D$. 

We define a finite set $h(G) \subseteq \IN$ and a function $s_G \colon h(G) \to \IQ_>$ for each non-empty $G\subseteq D$  as follows. If $\nu(U_G) = \nu(V_G)$ then let $h(G) = \emptyset$ and let $s_G$ be the empty function. Otherwise, the set 
\[C = \{ c\in\IN \mid (\forall b\in D)\, [b\prec c \iff  b \in G]\}\]
is non-empty because $\nu(U_G) > \nu(V_G)$ implies there is some ideal containing $G$ which is not in $V_G$. If there is some $c \in C$ with $\nu([c]_{\prec}) = \infty$, then set $h(G) = \{c\}$ and define $s_G \colon h(G) \to \IQ_>$ as $s_G(c) = M$. If no such $c\in C$ exists, then let $(c_i)_{i\in\IN}$ be an enumeration of $C$ and define
\[p_i = \nu([c_i]_{\prec}) - \nu\left([c_i]_{\prec} \cap \left(\bigcup_{k<i}[c_k]_{\prec}\cup V_G\right)\right).\]
Using modularity and a simple inductive argument, we have
\begin{eqnarray*}
\sum_{i\leq n} p_i &=& \nu(\bigcup_{i\leq n} [c_i]_{\prec}) - \nu\left(\bigcup_{i\leq n }[c_i]_{\prec}\cap V_G\right)\\
 &=& \nu(\bigcup_{i\leq n} [c_i]_{\prec} \cup V_G) - \nu(V_G)
\end{eqnarray*}
for each $n\in\IN$. Since $U_G = \bigcup_{i\in\IN} [c_i]_{\prec} \cup V_G$ and $\nu$ is Scott-continuous, there is $n_0\in\IN$ with 
\[ \left(\frac{1+\alpha}{2} \right)\sum_{i\leq n_0} p_i \geq \alpha(\nu(U_G) - \nu(V_G))\]
if $\nu(U_G) < \infty$, and
\[ \left(\frac{1+\alpha}{2} \right)\sum_{i\leq n_0} p_i \geq M\]
if $\nu(U_G) = \infty$. Define 
\[h(G) = \{ c_i \mid i\leq n_0 \,\&\, p_i >0 \}\]
and define $s_G \colon h(G) \to \IQ_>$ so that $s_G(c_i)$ is a positive rational satisfying
\[\left(\frac{1+\alpha}{2} \right) p_i \leq s_G(c_i) < p_i.\]

Since $h(G) \cap h(G') \not=\emptyset$ implies $G = G'$, there is $s\in\calB$ with 
\[dom(s) = \bigcup\{ h(G) \mid G \subseteq D\}\]
satisfying $s(c) = s_G(c)$ for the unique $G\subseteq D$ with  $c \in h(G)$. From the construction of $s$, if $F \subseteq h(G)$ is non-empty then
\begin{eqnarray}\label{form:sum}
\sum_{c\in F} s(c) < \nu(\bigcup_{c\in F} [c]_{\prec}) - \nu(\bigcup_{c\in F} [c]_{\prec} \cap V_G).
\end{eqnarray}
Furthermore, if $h(G)\not=\emptyset$, then $\nu(U_G) < \infty$ implies
\begin{eqnarray}\label{form:sum2}
\alpha\left( \nu(U_G) - \nu(V_G) \right) \leq \sum_{c\in h(G)} s(c),
\end{eqnarray}
and $\nu(U_G) = \infty$ implies
\begin{eqnarray}\label{form:sum3}
M \leq \sum_{c\in h(G)} s(c).
\end{eqnarray}

To show $s \in f_V(\nu)$, we must prove $\sum_{c\in F} s(c) < \nu(\bigcup_{c\in F} [c]_{\prec})$ for each non-empty $F\subseteq dom(s)$. This clearly holds when $F = \{c\}$ is a singleton. Next, assume it holds for all sets of size less than or equal to $n$, and let $F$ be a set of size $n+1$. We can assume $\nu(\bigcup_{c\in F} [c]_{\prec}) < \infty$, since otherwise the claim is trivial. Let $G\subseteq D$ be a set of minimal size satisfying $F\cap h(G) \not=\emptyset$. This implies that either $F\setminus h(G)$ is empty or else it satisfies the induction hypothesis. Furthermore, for any $c \in F \setminus h(G)$ there is $G' \subseteq D$ with $c \in h(G')$, and since the minimality of $G$ implies $G' \not\subseteq G$, there is $b \in G' \setminus G$ with $b \prec c$, which implies $U_G \cap [c]_{\prec} \subseteq V_G$. Therefore,
\begin{eqnarray*}
\sum_{c\in F} s(c) &=& \sum_{c\in F\cap h(G)}s_G(c) +  \sum_{c\in F\setminus h(G)} s(c)\\
&<& \nu(\bigcup_{c\in F\cap h(G)} [c]_{\prec}) - \nu(\bigcup_{c\in F\cap h(G)} [c]_{\prec} \cap V_G) + \nu(\bigcup_{c \in F\setminus h(G)}[c]_{\prec})\\
&&\text{ (by (\ref{form:sum}) and the induction hypothesis)}\\
&\leq&  \nu(\bigcup_{c\in F\cap h(G)} [c]_{\prec}) - \nu(\bigcup_{c\in F\cap h(G)} [c]_{\prec} \cap \bigcup_{c \in F\setminus h(G)}[c]_{\prec})\\
&&  + \nu(\bigcup_{c \in F\setminus h(G)}[c]_{\prec}) \\
&&\text{ (because $U_G \cap [c]_{\prec} \subseteq V_G$ for each $c \in F\setminus h(G)$)}\\
&=& \nu(\bigcup_{c\in F} [c]_{\prec}),
\end{eqnarray*}
which proves $s \in f_V(\nu)$.

Finally, we must show $r_0 \prec_V s$ and $r_1 \prec_V s$. Fix $i\in\{0,1\}$ and non-empty $F \subseteq dom(r_i)$. Set $P =  \{G \subseteq D \mid G \cap F \not=\emptyset\}$ and note that $\uparw F \cap dom(s) = \bigcup_{G \in P} h(G)$. If $\nu(U_G) < \infty$ for each $G \in P$, then using (\ref{form:sum2}) and the fact that $G \not=G'$ implies $h(G) \cap h(G') =\emptyset$, we have
\begin{eqnarray*}
\sum_{c\in\uparw F\cap dom(s)}s(c) &\geq& \sum_{G\in P} \alpha(\nu(U_G) - \nu(V_G))\\
&=& \alpha \nu\left(\bigcup_{G\in P} U_G \right) \qquad\text{ (by Lemma~\ref{lem:sum_of_differences})}\\
&=& \alpha \nu\left(\bigcup_{b\in F} [b]_{\prec}\right)\\
&>& \sum_{b\in F} r_i(b).
\end{eqnarray*}
Otherwise, there is $G\in P$ with $\nu(U_G) = \infty$, so (\ref{form:sum3}) implies
\[\sum_{c\in\uparw F\cap dom(s)}s(c) \geq M > \sum_{b\in F} r_i(b). \]
This completes the proof that $f_V(\nu)$ is directed.
\end{enumerate}
It only remains to show that $f_V$ is continuous. Fix $r \in\calB$. For each $F\subseteq dom(r)$ define $W_F = \bigcup_{b\in F} [b]_{\prec}$ and $q_F = \sum_{b\in F} r(b)$, and set $D = \{ F\subseteq dom(r) \mid F\not= \emptyset\}$. Then $f_V(\nu) \in [r]_{\prec_V}$ if and only if 
\[\nu \in \bigcap_{F\in D}  \langle W_F, q_F \rangle,\]
hence $f_V$ is continuous.
\end{proof}

%%%%%%%%%%%%%%%%%%%%%%%%%
%%%%%%%%%%%%%%%%%%%%%%%%%

\begin{lemma}
$g_V$ is well-defined and continuous.
\end{lemma}
\begin{proof}
We first show that $\nu = g_V(I)$ is a valuation for each $I\in \I{\prec_V}$.
\begin{enumerate}
\item
$\nu(\emptyset)=0$: Assume $U \in\PO(\I{\prec})$ and $\nu(U)>0$. Then there is $r_0\in I$ and $b_0\in dom(r_0)$ such that $[b_0]_{\prec} \subseteq U$ and $0 < r_0(b_0)$. Since $I$ is directed, there is an infinite sequence $r_0 \prec_V r_1 \prec_V \cdots$ in $I$. Since $b_0 \in dom(r_0)$ and $r_0 \prec_V r_1$, there is $b_1 \in dom(r_1)$ with $b_0 \prec b_1$. Similarly, there must be $b_2 \in dom(r_2)$ with $b_1 \prec b_2$. This yields an infinite sequence $b_0 \prec b_1 \prec \cdots$, hence $\{ c \in \IN \mid (\exists i\in \IN)\, c \prec b_i\}$ is an element of $[b_0]_{\prec} \subseteq U$. Therefore, $U\not=\emptyset$.
\item
$\nu(U)+\nu(V) = \nu(U \cup V) + \nu(U\cap V)$: We first show $\nu(U)+\nu(V) \leq \nu(U \cup V) + \nu(U\cap V)$. Let $r,s\in I$ be such that $(\forall b\in dom(r))\, [b]_{\prec}\subseteq U$ and $(\forall b\in dom(s))\, [b]_{\prec}\subseteq V$. Set
\[p_r = \sum_{b\in dom(r)} r(b), \qquad p_s = \sum_{b\in dom(s)} s(b).\]
Let $t\in I$ be a $\prec_V$-upper bound of $r$ and $s$. Let 
\begin{eqnarray*}
D_r &=& \{ c\in dom(t) \mid (\exists b\in dom(r))\, b \prec c \},\\
D_s &=& \{ c\in dom(t) \mid (\exists b\in dom(s))\, b \prec c \}.
\end{eqnarray*}
Note that $c \in D_r\cap D_s$ implies $[c]_{\prec} \subseteq U\cap V$. Set
\[q_0 = \sum_{c\in D_r\setminus D_s} t(c), \qquad q_1 = \sum_{c\in D_s\setminus D_r} t(c), \qquad q_2 = \sum_{c\in D_r\cap D_s} t(c).\]
Then $r \prec_V t$ implies $p_r \leq q_0 + q_2$ and $s \prec_V t$ implies $p_s \leq q_1 + q_2$. Furthermore, using the fact $t\in I$, Lemma~\ref{lem:formval_restrictions}, and the definition of $\nu$, we obtain $\nu(U\cup V) \geq q_0+q_1+q_2$ and $\nu(U\cap V) \geq  q_2$, hence $p_r + p_s \leq \nu(U\cup V) + \nu(U\cap V)$. It follows that $\nu(U)+\nu(V) \leq \nu(U \cup V) + \nu(U\cap V)$.

Next we show  $\nu(U \cup V) + \nu(U\cap V) \leq \nu(U)+\nu(V)$. Let $r,s\in I$ be such that $(\forall b\in dom(r))\, [b]_{\prec}\subseteq U\cup V$ and $(\forall b\in dom(s))\, [b]_{\prec}\subseteq U\cap V$. Let $K \subseteq \bigcup_{b\in dom(r)} [b]_{\prec}$ be as in Lemma~\ref{lem:formval_upperpower2}. Since $K$ is compact and $K \subseteq U\cup V$, there exists a finite set $F\subseteq \IN$ with $K \subseteq \bigcup_{b\in F} [b]_{\prec}$ and such that each $b\in F$ satisfies $[b]_{\prec} \subseteq U$ or $[b]_{\prec}\subseteq V$. Apply Lemma~\ref{lem:formval_upperpower2} to get $t\in I$ with $r \prec_V t$ and $F \prec_U dom(t)$ and $K \subseteq \bigcup_{c\in dom(t)} [c]_{\prec} \subseteq \bigcup_{b\in dom(r)} [b]_{\prec}$. Next let $u\in I$ be a $\prec_V$-upper bound of $t$ and $s$. By restricting the domain of $u$ if necessary, we can assume that $(dom(t) \cup dom(s)) \prec_U dom(u)$, hence every $c\in dom(u)$ satisfies $[c]_{\prec} \subseteq U$ or $[c]_{\prec} \subseteq V$. Let $u_0$ be the restriction of $u$ to have domain $dom(u_0) = \{ b\in dom(u) \mid [b]_{\prec} \subseteq U\}$, and let  $u_1$ be the restriction of $u$ to have domain $dom(u_1) = \{ b\in dom(u) \mid [b]_{\prec} \subseteq V\}$. Note that $u_0$ and $u_1$ are both in $I$ by Lemma~\ref{lem:formval_restrictions}, and that $dom(u) = dom(u_0)\cup dom(u_1)$. Then using the fact that $r\prec_V u$ and $s\prec_V u$, we have
\begin{eqnarray*}
\sum_{b\in dom(r)} r(b) +  \sum_{b\in dom(s)} s(b) &\leq& \sum_{c\in dom(u)} u(c) + \sum_{c\in dom(u_0) \cap dom(u_1) } u(c) \\
&=& \sum_{c\in dom(u_0)} u_0(c) + \sum_{c\in dom(u_1) } u_1(c)\\
&\leq& \nu(U) + \nu(V).
\end{eqnarray*}
Therefore, $\nu(U \cup V) + \nu(U\cap V) \leq \nu(U)+\nu(V)$.
\item
$\nu$ is a continuous function: Assume $U\in\PO(\I{\prec})$ and $q\in\IQ_{>0}$ and $\nu(U) > q$. Since $\I{\prec}$ is consonant (see \cite{DK19}), it suffices to find $K\in \PU(\I{\prec})$ such that $K\subseteq U$ and $\nu(W)>q$ whenever $W$ is an open set containing $K$. By definition of $g_V(I)$, there must be $r\in I$ such that $(\forall b\in dom(r))\, [b]_{\prec}\subseteq U$ and $\sum_{b\in dom(r)} r(b) > q$. Now let $K\in\PU(\I{\prec})$ be as in Lemma~\ref{lem:formval_upperpower2}. Then $K\subseteq U$, and if $K\subseteq W$ then there is $s \in I$ with  $r \prec_V s$ and $K \subseteq \bigcup_{c\in dom(s)} [c]_{\prec} \subseteq W$, hence $q < \sum_{c\in dom(s)} s(c) \leq \nu(W)$.
\end{enumerate}
It only remains to show that $g_V$ is continuous. Assume $g_V(I) \in \langle U, q\rangle$. Then there is $r\in I$ satisfying $(\forall b\in dom(r))\, [b]_{\prec}\subseteq U$ and $q < \sum_{b\in dom(r)} r(b)$. Then $I \in [r]_{\prec_V} \subseteq g_V^{-1}(\langle U,q\rangle)$, hence $g_V$ is continuous.
\end{proof}

\begin{theorem}\label{thrm:mainresult}
$\PV(\I{\prec})$ and $\I{\prec_V}$ are homeomorphic (via $f_V$ and $g_V$). 
\end{theorem}
\begin{proof}
It only remains to show that $f_V$ and $g_V$ are inverses of each other.

To show that $g_V\circ f_V$ is the identity function, it suffices to show that $g_V(f_V(\nu)) \in \langle U, q\rangle$ if and only if $\nu \in \langle U,q\rangle$ for each $\nu\in\PV(\I{\prec})$ and each subbasic open $\langle U,q\rangle$. If $g_V(f_V(\nu)) \in \langle U, q\rangle$, then there must be $r\in f_V(\nu)$ with $q < \sum_{b\in dom(r)}r(b)$ and $\bigcup_{b\in dom(r)}[b]_{\prec} \subseteq U$. This implies that $dom(r)\not=\emptyset$, and using the definition of $f_V$ we obtain $q < \sum_{b\in dom(r)}r(b) < \nu(\bigcup_{b\in dom(r)}[b]_{\prec}) \leq \nu(U)$, hence $\nu \in \langle U,q\rangle$. Conversely, if $\nu \in \langle U,q\rangle$ then since $\nu$ is continuous there exist $b_0, \ldots, b_n \in \IN$ such that $\bigcup_{i\leq n}[b_i]_{\prec} \subseteq U$ and $q < \nu(\bigcup_{i\leq n}[b_i]_{\prec})$. If $\nu([b_i]_{\prec}) = \infty$ for some $i\leq n$, then the partial function $r$ defined as $dom(r) = \{b_i\}$ and $r(b_i) = q+1$ is in $f_V(\nu)$, which implies $g_V(f_V(\nu)) \in \langle U, q\rangle$. Otherwise $\nu([b_i]_{\prec}) < \infty$ for each $i\leq n$, so define
\[ m_i = \nu([b_i]_{\prec}) - \nu\big([b_i]_{\prec} \cap \bigcup_{j<i} [b_j]_{\prec}\big).\]
Note that the modularity of $\nu$ implies $m_i = \nu(\bigcup_{j\leq i} [b_j]_{\prec}) -  \nu(\bigcup_{j<i} [b_j]_{\prec})$, hence a simple inductive argument yields $\sum_{i\leq n} m_i = \nu(\bigcup_{i\leq n} [b_i]_{\prec})$, which is strictly larger than $q$. Let $G = \{ i \mid m_i > 0\}$. Then there exists $r\in\calB$ with $dom(r) = \{ b_i \mid i \in G\}$ and $(\forall i\in G)\, r(b_i) < m_i$ and $q < \sum_{b\in dom(r)} r(b)$. If $F \subseteq G$ is non-empty, then
\begin{eqnarray*}
\sum_{i\in F} r(b_i) < \sum_{i\in F} m_i
&=& \sum_{i\in F} \left(\nu([b_i]_{\prec}) - \nu\big([b_i]_{\prec} \cap \bigcup_{j<i} [b_j]_{\prec}\big)\right)\\
&\leq& \sum_{i\in F} \left(\nu([b_i]_{\prec}) - \nu\big([b_i]_{\prec} \cap \bigcup_{\substack{j<i\\j\in F}} [b_j]_{\prec}\big)\right)\\
&=& \nu\left(\bigcup_{i\in F} [b_i]_{\prec}\right).
\end{eqnarray*}
Thus, $r \in f_V(\nu)$ and $q<\sum_{b\in dom(r)} r(b)$, hence $g_V(f_V(\nu)) \in \langle U, q\rangle$.

Next we show that $f_V(g_V(I)) = I$ for each $I \in \I{\prec_V}$. By unwinding the definitions of $f_V$ and $g_V$, we have $r \in f_V(g_V(I))$ if and only if for every non-empty $F \subseteq dom(r)$ there is $s \in I$ such that $\bigcup_{c\in dom(s)} [c]_{\prec} \subseteq \bigcup_{b\in F} [b]_{\prec}$ and $\sum_{b\in F} r(b) < \sum_{c\in dom(s)} s(c)$. Thus, given any $r \in I$, by Lemma~\ref{lem:formval_upperpower} there is $s\in I$ with $r \prec_V s$ and $dom(r) \prec_U dom(s)$, hence $\bigcup_{c\in dom(s)} [c]_{\prec} \subseteq \bigcup_{b\in F} [b]_{\prec}$ and $\sum_{b\in F} r(b) < \sum_{c\in dom(s)} s(c)$, which implies $r \in f_V(g_V(I))$. Therefore, $I \subseteq f_V(g_V(I))$.

To prove $f_V(g_V(I)) \subseteq I$, fix any $r \in f_V(g_V(I))$. Then for every non-empty $F \subseteq dom(r)$ there is $s_F \in I$ such that $\bigcup_{c\in dom(s_F)} [c]_{\prec} \subseteq \bigcup_{b\in F} [b]_{\prec}$ and $\sum_{b\in F} r(b) < \sum_{c\in dom(s_F)} s_F(c)$. Using Lemma~\ref{lem:formval_upperpower2}, we can assume that $F \prec_U dom(s_F)$. Let $s \in I$ be a $\prec_V$-upper bound of all of the $s_F$. Then for any non-empty $F \subseteq dom(r)$, we have
\begin{eqnarray*}
\sum_{b\in F} r(b) &<& \sum_{c\in  \uparw F \cap dom(s_F)} s_F(c) \qquad\text{ (by choice of $s_F$)}\\
&<& \sum_{c\in \uparw F \cap dom(s)} s(c) \qquad\text{ (because $s_F \prec_V s$ and $\prec$ is transitive).}
\end{eqnarray*}
Therefore $r \prec_V s$, hence $r \in I$ because $I$ is a lower-set. It follows that $f_V(g_V(I)) \subseteq I$, which completes the proof that $f_V(g_V(I)) = I$.
\end{proof}

We remark that the homeomorphisms $f_V$ and $g_V$ are computable in the sense of TTE \cite{W00} when $\prec$ is computably enumerable, and therefore our approach is consistent with previous work on computable measures in \cite{Sch:Prob,HR09,PSZ20}. The computability of $f_V$ is obvious. For $g_V$, note that for any $U \in \PO(\I{\prec})$ and any $A\subseteq \IN$ satisfying $U = \bigcup_{a \in A} [a]_{\prec}$, Lemma~\ref{lem:formval_upperpower2}  implies  
\[
g_V(I)(U) = \bigvee\left\{ \sum_{c\in dom(s)} s(c) \,\middle|\, s\in I \,\&\, (\forall c \in dom(s))(\exists a\in A)\,  a \prec c \right\},\]
which shows that $g_V$ is computable.

%%%%%%%%%%%%%%%%%%%%%%%%%%%%%%%%%%%%%%%%%%%%%%%%%%
%%%%%%%%%%%%%%%%%%%%%%%%%%%%%%%%%%%%%%%%%%%%%%%%%%
%%%%%%%%%%%%%%%%%%%%%%%%%%%%%%%%%%%%%%%%%%%%%%%%%%

%%
%% Bibliography
%%

\bibliographystyle{plainurl}
\bibliography{myrefs}

\end{document}